\documentclass[11pt]{article}
\usepackage
{
        amssymb,
        amsthm,
        amsmath,
        wasysym,
        graphicx,
        fullpage,
        enumerate,
        hhline
}

\newcommand {\set}   [1] {\left\{ #1 \right\}}
\newcommand {\brc}   [1] {\left(#1\right)}

\newcommand {\Exp}       {\mathbb{E}}
\newcommand {\Prob}  [1] {\Pr \brc{#1 }}

\newcommand {\E}     [1] {\Exp\left[#1\right]}
\newcommand {\EE}    [2] {\Exp_{#1}\left[#2\right]}

\newcommand {\iprod} [2] {\langle #1, #2 \rangle}

\newcommand {\calS}   {{\cal{S}}}

\newcommand {\ONE}      {\text{\textbf{1}}}

\newtheorem{theorem}{Theorem}[section]
\newtheorem{lemma}[theorem]{Lemma}

\newtheorem{corollary}[theorem]{Corollary}
\newtheorem{definition}[theorem]{Definition}

\title{How to Play Unique Games on Expanders}
\author{Konstantin Makarychev\footnote{IBM T.J. Watson Research Center, Yorktown Heights, NY 10598.}
\and Yury Makarychev\footnote{Microsoft Research New England, One Memorial Drive, Cambridge, MA 02142.}}
\date{}
\begin{document}
\maketitle
\begin{abstract}
In this note we improve a recent result by Arora, Khot, Kolla, Steurer,
Tulsiani, and Vishnoi on solving the Unique Games problem on expanders.

Given a $(1-\varepsilon)$-satisfiable instance of Unique Games with the constraint graph $G$,
our algorithm finds an assignment satisfying at least a  $1- C \varepsilon/h_G$ 
fraction of all constraints if
$\varepsilon < c \lambda_G$ where $h_G$ is the edge expansion of $G$, $\lambda_G$ is the 
second smallest eigenvalue of the Laplacian of $G$, and $C$ and $c$ are some absolute constants.

We refer the reader to~\cite{AKK, CMM, CMM1, GT, Tre05} for the motivation and an overview of related work.
\end{abstract}

\section{Preliminaries: Expanders, Unique Games and SDP}
\subsection{Unique Games and Expanders}
In this note we study the Unique Games problem on regular expanders.

\begin{definition}[Unique Games Problem]
Given a constraint graph $G = (V,E)$
and a set of permutations
$\pi_{uv}$ on the set $[k] = \set{1,\dots, k}$ (for all edges $(u,v)$),
the goal is to assign a value (state) $x_u$ from $[k]$ to each vertex $u$
so as to satisfy the maximum number of constraints of the form
$\pi_{uv}(x_u) = x_v$. The cost of a solution is the fraction
of satisfied constraints.
\end{definition}

We assume that the underlying graph $G=(V,E)$ is a $d$-regular expander. The two key
parameters of the expander $G$ are the edge expansion $h_G$ and
the second eigenvalue of the Laplacian $\lambda_G$. The edge expansion
gives a lower bound on the size of every cut: for every subset of vertices
$X\subset V$, the size of the cut between $X$ and $|V\setminus X|$ is at least
$$h_G \times \frac{\min (|X|, |V\setminus X|)}{|V|} |E|.$$
It is formally defined as follows:
$$h_G =
\min_{X \subset V}
\left(\frac{|\delta(X,V\setminus X)|}{|E|} \left/
\frac{\min (|X|, |V\setminus X|)}{|V|}\right. \right),$$
here $\delta(X,V\setminus X)$ denotes the cut --- the set of edges going from $X$ to $V\setminus
X$. One can think of the second eigenvalue of the Laplacian
$$L_G(u,v)  =
\begin{cases}
1, & \text{if } u=v\\
-1/d, & \text{if } (u,v) \in E\\
0, &\text{otherwise.}
\end{cases}
$$
as of continuous
relaxation of the edge expansion. Note that the smallest eigenvalue of $L_G$ is 0; and
the corresponding eigenvector is a vector of all 1's, denoted by~$\ONE$. Thus
$$\lambda_G = \min_{x \perp \text{\large{\textbf{1}}}} \frac{\iprod{x}{L_Gx}}{\|x\|^2}.$$
Cheeger's inequality,
$$h_G^2/8 \leq \lambda_G \leq h_G,$$
shows that $h_G$ and $\lambda_G$ are closely related; however $\lambda_G$ can be much smaller than $h_G$ (the lower bound
in the inequality is tight).

\subsection{Results of Arora, Khot, Kolla, Steurer, Tulsiani, and Vishnoi}
In a recent work~\cite{AKK}, Arora, Khot, Kolla, Steurer, Tulsiani, and Vishnoi
showed how given a $(1-\varepsilon)$ satisfiable instance of Unique Games
(i.e. an instance in which the
optimal solution satisfies at least a $(1-\varepsilon)$ fraction of
constraints), one can obtain a solution of cost
$$1- C \frac{\varepsilon}{\lambda_G} \log \left(\frac{\lambda_G}{\varepsilon}\right)$$
in polynomial time, here $C$ is an absolute constant. We improve their result and show that, if
the ratio $\varepsilon/\lambda_G$ is less than some universal positive constant $c$, one
can obtain a solution of cost
$$1- C' \frac{\varepsilon}{h_G}$$
in polynomial time.
As mentioned above, $\lambda_G$ can be significantly smaller than $h_G$, then
our result gives much better approximation guarantee. However, even if $\lambda_G \approx h_G$,
our bound is asymptotically stronger, since
$$1- C' \frac{\varepsilon}{h_G} \geq 1- C' \frac{\varepsilon}{\lambda_G}$$
(our bound does not have a $\log (\lambda_G/\varepsilon)$ factor).
It is an interesting open question,
if one can replace the condition $\varepsilon/\lambda_G < c$ with $\varepsilon/h_G < c$.

\subsection{Semidefinite Relaxation}
We use the standard SDP relaxation for the Unique Games problem.

$$
   \text{minimize }
       \frac{1}{2|E|} \sum_{(u,v)\in E} \sum_{i=1}^{k} \|u_{i} - v_{\pi_{uv}(i)}\|^2
$$
subject to
\begin{eqnarray}
\forall u\in V \; \forall i, j \in [k], i\neq j &\;& \iprod{u_i}{u_j}  = 0 \\
\forall u\in V &\;& \sum_{i=1}^{k}\|u_i\|^2 = 1 \label{cond:SumOne}\\
\forall u,v, w\in V\;\forall i,j,l \in [k]  &\;&  \|u_i - w_l\|^2 \leq
\|u_i-v_j\|^2 + \|v_j - w_l \|^2 \label{eq:Triangle}\\
\forall u,v \in V\;\forall i,j \in [k]
  &\;&  \|u_i - v_j\|^2 \leq \|u_i\|^2 + \|v_j \|^2 \label{eq:Triangle0}\\
\forall u,v \in V\;\forall i,j \in [k]
  &\;&  \|u_i\|^2 \leq \|u_i - v_j\|^2 + \|v_j \|^2 \label{eq:Triangle1}
\end{eqnarray}

For every vertex $u$ and state $i$ we introduce a vector $u_i$. In the intended integral
solution $u_i = 1$, if $u$ has state $i$; and $u_i =0$, otherwise. All SDP constraints
are satisfied in the integral solution; thus this is a valid relaxation. The objective function
of the SDP measures what fraction of all Unique Games constraints is \textit{not satisfied}.

\section{Algorithm}
We define the \textit{earthmover distance} between two sets of orthogonal vectors $\set{u_1,\dots, u_k}$
and $\set{v_1,\dots, v_k}$ as follows:
$$
\Delta(\set{u}_i, \set{v}_i) \equiv
\min_{\sigma(i) \in \calS_k} \sum_{i=1}^{k} \|u_i - v_{\sigma(i)}\|^2,
$$
here $\calS_k$ is the symmetric group, the group of all permutations on the set $[k]=\set{1,\dots,k}$. Given an SDP solution $\set{u_i}_{u,i}$ we define the earthmover distance between vertices in a natural way:
$$
\Delta(u,v) = \Delta(\set{u_1,\dots, u_k},\set{v_1,\dots, v_k}).
$$

Arora et al.~\cite{AKK} proved that if an instance of Unique Games on an expander is almost satisfiable,
then the average earthmover distance between two vertices (defined by the SDP solution) is small. We will need the following corollary from their results:

\textit{For every $R\in(0,1)$, there exists a positive $c$, such that for every
$(1-\varepsilon)$ satisfiable instance of Unique Games on an expander graph
$G$, if $\varepsilon/\lambda_G < c$, then the
expected earthmover distance between two random vertices
is less than~$R$ i.e.
$$\EE{u,v\in V}{\Delta (u, v)}\leq R.$$}
In fact, Arora et al.~\cite{AKK} showed that $c \geq \Omega(R / \log (1/R))$, but we
will not use this bound. Moreover, in the rest of the paper, we fix the value of $R < 1/4$. We
pick $c_R$, so that if $\varepsilon/\lambda_G < c_R$, then
\begin{equation}\label{eq:R4}
\EE{u,v\in V}{\Delta (u, v)}\leq R/4.
\end{equation}

Our algorithm transforms vectors $\set{u_i}_{u,i}$ in the SDP solution to vectors $\set{\tilde{u}_i}_{u,i}$
using a \textit{normalization} technique introduced by Chlamtac, Makarychev and Makarychev~\cite{CMM1}:
\begin{lemma}\label{lem:CMM1}\cite{CMM1} For every SDP solution $\set{u_i}_{u,i}$, there exists a set of vectors
$\set{\tilde{u}_i}_{u,i}$ satisfying the following properties:
\begin{enumerate}
\item Triangle inequalities in $\ell_2^2$: for all vertices $u$, $v$, $w$ in $V$ and
all states $i$, $p$, $q$ in $[k],$
$$
\|\tilde{u}_i - \tilde{v}_p\|_2^2 + \|\tilde{v}_p
- \tilde{w}_q\|_2^2 \geq \|\tilde{u}_i - \tilde{w}_q\|_2^2.$$
\item For all vertices $u, v$ in $V$ and all states $i,j$ in $[k]$,
$$\iprod{\tilde{u}_i}{\tilde{v}_j} = \frac{\iprod{u_i}{v_j}}{\max(\|u_i\|^2,\|v_j\|^2)}.$$
\item For all non-zero vectors $u_i$, $\|\tilde{u}_i\|_2^2 = 1$.
\item For all $u$ in $V$ and $i\neq j$ in $[k]$, the vectors
$\tilde{u}_i$ and $\tilde{u}_j$ are  orthogonal.
\item For all  $u$ and $v$ in $V$ and $i$ and $j$ in $[k]$,
$$\|\tilde{v}_j - \tilde{u}_i\|_2^2 \leq \frac{2\,\|v_j - u_i \|^2}{\max(\|u_i\|^2, \|v_j\|^2)}.$$
\end{enumerate}
The set of vectors $\set{\tilde{u}_i}_{u,i}$ can be obtained in polynomial time.
\end{lemma}

Now we are ready to describe the rounding algorithm. The algorithm
given an SDP solution, outputs an assignment of states (labels)
to the vertices.

%%\pagebreak
%%\begin{figure*}[t]
\rule{0pt}{12pt}
\hrule height 0.8pt
\rule{0pt}{1pt}
\hrule height 0.4pt
\rule{0pt}{3pt}

\textbf{Approximation Algorithm}

\begin{enumerate}
\item[] \textbf{Input:} an SDP solution $\set{u_i}_{u,i}$ of cost $\varepsilon$.
\item[] \textbf{Initialization}
\item Pick a random vertex $u$ (uniformly distributed) in $V$. We call this vertex \textit{the initial vertex}.
\item Pick a random state $i \in [k]$ for $u$; choose state $i$ with probability $\|u_i\|^2$. Note that $\|u_1\|^2 + \dots + \|u_k\|^2 = 1$. We call $i$ \textit{the initial state}.
\item Pick a random number $t$ uniformly distributed in the segment $[0,\|u_i\|^2]$.
\item Pick a random $r$ in $[R,2R]$.
\item[] \textbf{Normalization}
\item Obtain vectors $\set{\tilde{u}_i}_{u,i}$ as in Lemma~\ref{lem:CMM1}.
\item[] \textbf{Propagation}
\item For every vertex $v$,
\begin{itemize}
\item  Find all states $p \in [k]$ such that $\|v_p\|^2 \geq t$ and
$\|\tilde{v}_p - \tilde{u}_i\|^2 \leq r$. Denote the set of $p$'s by~$S_v$:
$$S_v = \set{p:  \|v_p\|^2 \geq t \text { and } \|\tilde{v}_p - \tilde{u}_i\|^2 \leq r}.$$
\item If $S_v$ contains exactly one element $p$, then assign the state $p$ to $v$.
\item Otherwise, assign an arbitrary (say, random) state to $v$.
\end{itemize}
\end{enumerate}

\hrule height 0.4pt
\rule{0pt}{1pt}
\hrule height 0.8pt
\rule{0pt}{12pt}
%%\end{figure*}

Denote by $\sigma_{vw}$ the partial mapping from $[k]$
to $[k]$ that maps $p$ to $q$ if $\|\tilde{v}_p-\tilde{w}_q\|^2\leq 4R$. Note
that $\sigma_{vw}$ is well defined i.e. $p$ cannot be mapped to different
states $q$ and $q'$: if  $\|\tilde{v}_p-\tilde{w}_q\|^2\leq 4R$ and
$\|\tilde{v}_p-\tilde{w}_{q'}\|^2\leq 4R$, then, by the $\ell_2^2$ triangle inequality (see Lemma~~\ref{lem:CMM1}(1)),
$\|\tilde{w}_q-\tilde{w}_{q'}\|^2\leq 8R$, but $\tilde{w}_q$ and $\tilde{w}_{q'}$ are
orthogonal unit vectors, so
$$\|\tilde{w}_q-\tilde{w}_{q'}\|^2 = 2 > 8R.$$
Clearly, $\sigma_{vw}$ defines a partial matching between states of $v$ and states of $w$:
if  $\sigma_{vw}(p) = q$, then  $\sigma_{wv}(q) = p$.

\begin{lemma} If $p\in S_v$ and $q\in S_w$ with non-zero probability,
then $q=\sigma_{vw}(p)$.
\end{lemma}
\begin{proof}
If $p\in S_v$ and $q\in S_w$ then for some vertex $u$ and state $i$,
$\|\tilde{v}_p - \tilde{u}_i\|^2 \leq 2R$
and
$\|\tilde{w}_q - \tilde{u}_i\|^2 \leq 2R$, thus by the triangle inequality
$\|\tilde{v}_p - \tilde{w}_q\|^2 \leq 4R$ and by the definition of $\sigma_{vw}$,
$q=\sigma_{vw}(p)$.
\end{proof}

\begin{corollary}\label{cor:oneelement}
Suppose, that $p\in S_v$, then the set $S_w$ either equals $\set{\sigma_{vw}(p)}$ or is empty (if $\sigma_{vw}(p)$ is not defined, then $S_w$ is empty). Particularly, if $u$ and~$i$ are the initial
vertex and state, then the set $S_w$ either equals $\set{\sigma_{uw}(i)}$ or is empty.
Thus, every set $S_w$ contains at most one element.
\end{corollary}

\begin{lemma}\label{lem:probui}
For every choice of the initial vertex $u$, for every $v\in V$ and $p\in[k]$
the probability that $p\in S_v$ is at most $\|v_p\|^2$.
\end{lemma}
\begin{proof}
If $p\in S_v$, then $i = \sigma_{vu} (p)$ is the initial state of $u$ and
$t \leq \|v_p\|^2$. The probability that both these events happen is
$$\Prob{i\in S_u}\times \Prob{t \leq \|v_p\|^2}=\|u_i\|^2 \times \min(\|v_p\|^2/\|u_i\|^2,1) \leq \|v_p\|^2$$
(recall that $t$ is a random real number on the segment $[0, \|u_i\|^2]$).
\end{proof}

Denote the set of those vertices $v$ for which $S_v$ contains exactly one element by $X$. First, we show that on average $X$ contains a constant fraction of all vertices (later we will prove a
much stronger bound on the size of $X$).

\begin{lemma}\label{lem:quart}
If $\varepsilon/\lambda_G \leq c_R$, then the expected size of $X$ is at least $|V|/4$.
\end{lemma}
\begin{proof}
Consider an arbitrary vertex $v$. Estimate the probability that
$p \in S_v$ given that $u$ is the initial vertex.
Suppose that there exists $q$ such that $\|v_p-u_q\|^2 \leq \|v_p\|^2\cdot R/2$, then
$$\|\tilde{u}_q - \tilde{v}_p\|^2 \leq \frac{2 \|u_q-v_p\|^2}{\max(\|u_q\|^2, \|v_p\|^2)} \leq R.$$
Thus, $q=\sigma_{vu}(p)$ and $\|\tilde{u}_q - \tilde{v}_p\|^2 \leq r$
with probability 1. Hence, if $q$ is chosen as the initial state and
$\|v_p\|^2 \geq t$, then $v_p \in S_v$. The probability
of this event is $\|u_q\|^2 \times \min(\|v_p\|^2/\|u_q\|^2,1)$. Notice that
$$\|u_q\|^2 \times \min(\|v_p\|^2/\|u_q\|^2,1) = \min(\|v_p\|^2,\|u_q\|^2)
\geq \|v_p\|^2 - \|u_q - v_p\|^2 \geq \frac{\|v_p\|^2}{2}.$$

Now, consider all $p$'s for which there exists $q$ such that $\|v_p-u_q\|^2 \leq \|v_p\|^2\cdot R/2$.
The probability that one of them belongs to $S_v$, and thus $v\in X$, is at least
\begin{eqnarray*}
\frac{1}{2}
\sum_{p:\min_q(\|u_q-v_p\|^2) \leq \|v_p\|^2 \cdot R/2} \|v_p\|^2 &=&
\frac{1}{2}\sum_{p=1}^k \|v_p\|^2 - \frac{1}{2}
\sum_{p: \min_q(\|u_q-v_p\|^2) > \|v_p\|^2 \cdot R/2} \|v_p\|^2 \\
&\geq&
\frac{1}{2} - \frac{1}{2}\times \sum_{p=1}^k\frac{2}{R} \min_q(\|u_q-v_p\|^2) \\
&\geq& \frac{1}{2} -  \frac{\Delta(\set{u}_q, \set{v}_p)}{R}.
\end{eqnarray*}
Since the average value of $\Delta(\set{u}_q, \set{v}_p)$ over all pairs $(u,v)$
is at most $R/4$ (see~(\ref{eq:R4})), the expected size of $X$ (for random initial vertex $u$)
is at least $|V|/4$.
\end{proof}
\begin{corollary}\label{cor:quart}
If $\varepsilon/\lambda_G \leq c_R$, then the size of $X$ is greater than $|V|/8$ with probability greater than
$1/8$.
\end{corollary}

\begin{lemma}\label{lem:expcut}
The expected size of the cut between $X$ and $V\setminus X$ is at most $6\varepsilon/R |E|$.
\end{lemma}
\begin{proof}
We show that the size of the cut between $X$ and $V\setminus X$ is at most $6\varepsilon/R|E|$ in the
expectation for any choice of the initial vertex $u$.
%%Fix an arbitrary vertex $u$.
%%For all vertices $v$, let us relabel all vectors so that $\tilde{v}_i$ is the closest vector to %%$\tilde{u}_i$ (this is done only for the sake of analysis; and clearly does not affect the behavior %%of the algorithm). It is easy to verify that if $\sigma_{uv}(i)$ is defined, then $i=\sigma_{uv}(i)$.
Fix an edge $(v,w)$ and estimate the probability that $v\in X$ and $w \in V\setminus X$.
If $v\in X$ and $w \in V\setminus X$, then $S_v$ contains a unique state $p$, but $S_w$ is empty (see
Corollary~\ref{cor:oneelement}) and, particularly, $\pi_{vw}(p)\notin S_w$. This happens in two cases:
\begin{itemize}
\item There exists $p$ such that $i=\sigma_{vu}(p)$ is the initial state of $u$ and $\|w_{\pi_{vw}(p)}\|^2 < t \leq \|v_p\|^2$. The probability of this event is at most
$$\sum_{p = 1}^k \|u_{\sigma_{vu}(p)}\|^2
\times\left|\frac{\|v_p\|^2 - \|w_{\pi_{vw}(p)}\|^2}{\|u_{\sigma_{vu}(p)}\|^2}\right|
\leq \sum_{p = 1}^k \|v_p - w_{\pi_{vw}(p)}\|^2.$$
\item There exists $p$ such that $i=\sigma_{vu}(p)$ is the initial state of $u$,
$t \leq \|v_p\|^2$ and
$ \|\tilde{u}_i - \tilde{v}_p\|^2 < r \leq
\|\tilde{u}_i - \tilde{w}_{\pi_{vw}(p)}\|^2$.
The probability of this event is at most
\begin{eqnarray*}
\sum_{p = 1}^k \|u_{\sigma_{vu}(p)}\|^2 \times \frac{\|v_p\|^2}{\|u_{\sigma_{vu}(p)}\|^2} &\times&
\left|\frac{\|\tilde{u}_{\sigma_{vu}(p)} - \tilde{w}_{\pi_{vw}(p)}\|^2 -
\|\tilde{u}_{\sigma_{vu}(p)} - \tilde{v}_p\|^2}{R}\right| \\
&\leq& \sum_{p = 1}^k \|v_p\|^2 \times \frac{\|\tilde{v}_p - \tilde{w}_{\pi_{vw}(p)}\|^2}{R}\\
&\leq& \sum_{p = 1}^k \|v_p\|^2 \times \frac{2\|v_p - w_{\pi_{vw}(p)}\|^2}{R\cdot
\max(\|v_p\|^2, \|w_{\pi_{vw}(p)}\|^2)}\\
&\leq& \frac{2}{R}\sum_{p = 1}^k \|v_p - w_{\pi_{vw}(p)}\|^2.
\end{eqnarray*}
\end{itemize}
Note that the probability of the first event is zero, if $\|w_{\pi_{vw}(p)}\|^2 \geq \|v_p\|^2$;
and the probability of the second event is zero, if
$ \|\tilde{u}_{\sigma_{vu}(p)} - \tilde{v}_p\|^2 \geq \|\tilde{u}_{\sigma_{vu}(p)} - \tilde{w}_{\pi_{vw}(p)}\|^2$.

Since the SDP value equals
$$\frac{1}{2|E|}\sum_{(v,w)\in E}\sum_{p = 1}^k \|v_p - w_{\pi_{vw}(p)}\|^2 \leq \varepsilon.$$
The expected fraction of cut edges is at most $6\varepsilon/R$.

\end{proof}

\begin{lemma}\label{lem:largeX}
If $\varepsilon \leq \min(c_R \lambda_G, h_G R/1000)$, then with probability at least 1/16 the size of $X$
is at least
$$\left(1 - \frac{100\varepsilon}{h_G R}\right) |V|.$$
\end{lemma}
\begin{proof}
The expected size of the cut $\delta(X, V\setminus X)$ between $X$ and $V \setminus X$ is less than $6\varepsilon/R |E|$. Hence, since the graph $G$ is an expander, one of the sets $X$ or $V\setminus X$ must be small:
$$\E{\min(|X|,|V\setminus X|)} \leq \frac{1}{h_G}\times \frac{\E{|\delta (X,V\setminus X)|}}{|E|}\times |V| \leq
\frac{6\varepsilon}{h_GR}|V|.$$
By Markov's Inequality,
$$\Prob{\min(|X|,|V\setminus X|) \leq \frac{100\varepsilon}{h_GR}|V|} \geq 1 - \frac{1}{16}.$$
Observe, that $ 100\varepsilon/(h_GR)|V| < |V|/8$. However, by Corollary~\ref{cor:quart}, the size of
$X$ is greater than $|V|/8$ with probability greater than $1/8$. Thus
$$\Prob{|V\setminus X| \leq \frac{100\varepsilon}{h_GR}|V|} \geq \frac{1}{16}.$$
\end{proof}

\begin{lemma}\label{lem:epsuv}
The probability that for an arbitrary edge $(v,w)$, the constraint between
$v$ and $w$ is not satisfied, but $v$ and $w$ are in $X$ is at most $4\varepsilon_{vw}$, where
$$\varepsilon_{vw} = \frac{1}{2}\sum_{i=1}^{k} \|v_i- w_{\pi_{vw}(i)}\|^2.$$
\end{lemma}
\begin{proof}
We show that for every choice of the initial vertex $u$ the desired probability is at most $4\varepsilon_{vw}$. Recall, that if $p \in S_v$ and $q \in S_w$, then $q=\sigma_{vw} (p)$.
The constraint between $v$ and $w$ is not satisfied if $q\neq\pi_{vw} (p)$.
Hence, the probability that the constraint is not satisfied is at most,
$$\sum_{p: \pi_{vw} (p) \neq \sigma_{vw} (p)} \Prob{p \in S_v}.$$
If $\pi_{vw} (p) \neq \sigma_{vw} (p)$, then
$$\|\tilde{v}_p-\tilde{w}_{\pi_{vw}(p)}\|^2 \geq
\|\tilde{w}_{\pi_{vw}(p)}-\tilde{w}_{\sigma_{vw}(p)}\|^2 - \|\tilde{v}_p-\tilde{w}_{\sigma_{vw}(p)}\|^2 \geq 2 - 4R \geq 1.$$
Hence, by Lemma~\ref{lem:CMM1} (5),
$$\|v_p-w_{\pi_{vw}(p)}\|^2 \geq \|v_p\|^2 /2.$$
Therefore, by Lemma~\ref{lem:probui},
$$\sum_{p: \pi_{vw} (p) \neq \sigma_{vw} (p)} \Prob{p \in S_v} \leq
\sum_{p: \pi_{vw} (p) \neq \sigma_{vw} (p)} \|v_p\|^2
\leq 2 \sum_{p=1}^{k}\|v_p-w_{\pi_{vw}(p)}\|^2 = 4\varepsilon_{vw}.$$
\end{proof}

\begin{theorem}
There exists a polynomial time approximation algorithm
that given a $(1-\varepsilon)$ satisfiable instance
of Unique Games on a $d$-expander graph $G$ with $\varepsilon/\lambda_G \leq c$,
the algorithm finds a solution of cost
$$1-C \frac{\varepsilon}{h_G},$$
where  $c$ and $C$ are some positive absolute constants.
\end{theorem}
\begin{proof}
We describe a randomized polynomial time algorithm. Our algorithm may return a solution to the SDP
or output a special value \textit{fail}. We show that the algorithm
outputs a solution with a constant probability (that is, the probability
of failure is bounded away from 1); and conditional on the event that
the algorithm outputs a solution its expected value is
\begin{equation}\label{eq:final}
1-C \frac{\varepsilon}{h_G}.
\end{equation}
Then we argue that the algorithm can be easily derandomized --- simply by
enumerating all possible values of the random variables used in the
algorithm and picking the best solution. Hence, the deterministic
algorithm finds a solution of cost at least~(\ref{eq:final}).

The randomized algorithm first solves the SDP and then runs the rounding procedure
described above. If the size of the set $X$ is more than
$$\left(1 - \frac{100\varepsilon}{h_G R}\right) |V|,$$
the algorithm outputs the obtained solution; otherwise, it outputs
\textit{fail}.

Let us analyze the algorithm. By Lemma~\ref{lem:largeX},
it succeeds with probability at least $1/16$. The fraction of edges having at least
one endpoint in $V\setminus X$ is at most
$100\varepsilon/(h_G R)$ (since the graph is $d$-regular). We conservatively
assume that the constraints corresponding to these edges are violated.
The expected number of violated constraints between vertices in $X$,
by Lemma~\ref{lem:epsuv} is at most
$$\frac{4\sum_{(u,v)\in E} \varepsilon_{uv}}{\Prob{|X|\geq 100\varepsilon/(h_G R)}}
\leq 64 \times \left(\frac{1}{2}\sum_{(u,v)\in E} \|u_i -v_{\pi_{vw}(i)}\|^2\right) \leq
64 \varepsilon |E|.
$$
The total fraction of violated constraints is at most
$100\varepsilon/(h_G R) + 64 \varepsilon$.
\end{proof}
%%%%%%%%%%%%%%%%%%%%%%%%%%%%%%%%%%%%%%%%%%%%%%%%%%%%%%%%%%%%%%%%%%%%%%%%
%% Bibliography
%%%%%%%%%%%%%%%%%%%%%%%%%%%%%%%%%%%%%%%%%%%%%%%%%%%%%%%%%%%%%%%%%%%%%%%%


\begin{thebibliography}{1W}
\bibitem{AKK}
S. Arora, S. Khot, A. Kolla, D. Steurer, M. Tulsiani, and N. Vishnoi.
\newblock \textit{Near-Optimal Algorithms for Unique Games.}
\newblock In Proceedings of the 40th ACM Symposium on Theory of Computing, pp. 21--28, 2008.
%%
\bibitem{CMM}
M.~Charikar, K.~Makarychev, and Y.~Makarychev.
\newblock \textit{Near-Optimal Algorithms for Unique Games.}
\newblock In Proceedings of the 38th ACM Symposium on Theory of Computing,
pp.~205--214, 2006.

\bibitem{CMM1}
E.~Chlamtac, K.~Makarychev, and Y.~Makarychev.
\newblock \textit{Unique games on expanding constraint graphs are easy.}
\newblock In Proceedings of the 47th IEEE Symposium on 
Foundations of Computer Science, pp.~687--696, 2006.
%%
\bibitem{GT}
A.~Gupta and K.~Talwar.
\newblock \textit{Approximating Unique Games.}
\newblock In Proceedings of the 17th ACM-SIAM Symposium on Discrete Algorithms,
pp.~99--106, 2006.
%%
\bibitem{Kho02} S. Khot.
\newblock \textit{On the power of unique 2-prover 1-round games}.
\newblock In Proceedings of the 34th ACM Symposium on Theory of Computing,
pp.~767--775, 2002.
%%

\bibitem{KKMO05} S.~Khot, G.~Kindler, E.~Mossel, and R.~O'Donnell.
\newblock{\textit{Optimal inapproximability results for MAX-CUT and other two-variable CSPs?}}
\newblock SIAM Journal of Computing 37(1), pp. 319--357, 2007.
%%
\bibitem{Tre05}
L. Trevisan.
\newblock \textit{Approximation Algorithms for Unique Games.}
\newblock  In Proceedings of the 46th IEEE Symposium on
Foundations of Computer Science, pp.~197--205, 2005.
\end{thebibliography}
\end{document}